\documentclass[a4paper,UKenglish,cleveref, autoref,unicode]{lipics-v2019}
\usepackage[linesnumbered,ruled, vlined]{algorithm2e}
\nolinenumbers
\title{Near-Linear Time Algorithms for Streett Objectives in Graphs and MDPs}

\author{Krishnendu Chatterjee}{IST Austria
}{krish.chat@ist.ac.at}{}{}
\author{Wolfgang Dvo\v{r}\'{a}k}{Institute of Logic and Computation, TU
	Wien
}{dvorak@dbai.tuwien.ac.at}{}{}

\author{Monika Henzinger}{Theory and Application of Algorithms,
	University of Vienna
}{monika.henzinger@univie.ac.at}{}{}

\author{Alexander Svozil}{Theory and Application of Algorithms,
	University of
	Vienna
}{alexander.svozil@univie.ac.at}{}{}

\authorrunning{K.\@ Chatterjee, W.\@ Dvo\v{r}\'{a}k, M.\@ Henzinger, A.\@ Svozil}
\Copyright{Krishnendu Chatterjee, Wolfgang Dvo\v{r}\'{a}k, Monika Henzinger, Alexander Svozil}

\ccsdesc[500]{Mathematics of computing~Combinatorial algorithms}
\ccsdesc[500]{Software and its engineering~Formal software verification}

\keywords{model checking, graph games, Streett games}

\category{}

\relatedversion{}

\supplement{}

\funding{A. S. is fully supported by the Vienna
	Science and Technology Fund (WWTF) through project ICT15-003. K.C. is supported by the
	Austrian Science
	Fund (FWF) NFN Grant No S11407-N23 (RiSE/SHiNE) and an ERC Starting
	grant (279307: Graph Games). For M.H the research leading to these
	results has received funding from the European Research Council under
	the European Union’s Seventh Framework Programme (FP/2007-2013) / ERC
	Grant Agreement no. 340506.}

\EventEditors{Wan Fokkink and Rob van Glabbeek}
\EventNoEds{2}
\EventLongTitle{30th International Conference on Concurrency Theory (CONCUR 2019)}
\EventShortTitle{CONCUR 2019}
\EventAcronym{CONCUR}
\EventYear{2019}
\EventDate{August 27--30, 2019}
\EventLocation{Amsterdam, the Netherlands}
\EventLogo{}
\SeriesVolume{140}
\ArticleNo{3}
\newcommand{\para}[1]{\smallskip\noindent\emph{#1}}

\newcommand{\SCCs}[1]{\mathtt{SCCs}(#1)}

\newcommand{\Streett}{\mathsf{Streett}}
\newcommand{\Remove}[1]{\mathtt{Remove}(#1)}
\newcommand{\Bad}[1]{\mathtt{Bad}(#1)}
\newcommand{\Construct}[1]{\mathtt{Construct}(#1)}
\newcommand{\add}[1]{\mathtt{add}(#1)}
\newcommand{\pull}{\mathtt{deque}()}
\newcommand{\ls}{\langle}
\newcommand{\rs}{\rangle}
\newcommand{\D}{\mathcal{D}}
\newcommand{\A}{\mathcal{A}}
\newcommand{\edgeset}[1]{E(#1)}
\newcommand{\In}[1]{\mathit{In}(#1)}
\newcommand{\Out}[1]{\mathit{Out}(#1)}
\newcommand{\GraphReach}[1]{\mathsf{GraphReach}(#1)}
\renewcommand{\O}{\widetilde{O}}
\newcommand{\bits}{\mathit{bits}}
\newcommand{\attr}[2]{\mathsf{attr_{#1}(#2)}}

\newcommand{\MEC}{\mathsf{MEC}}
\newcommand{\True}{\mathbf{True}}
\newcommand{\False}{\mathbf{False}}
\newcommand{\ASW}[1]{\langle\!\langle\text{1}\rangle\!\rangle_{a.s.}(#1)}
\newcommand{\Reach}[1]{\mathit{Reach}(#1)}
\newcommand{\delete}[1]{\mathsf{delete}(#1)}
\newcommand{\deleteannounce}[1]{\mathsf{deleteAnnounce}(#1)}
\newcommand{\query}[1]{\mathsf{query}(#1)}
\newcommand{\rep}[1]{\mathsf{rep}(#1)}
\newcommand{\deleteannouncenooutgoing}[1]{\mathsf{deleteAnnounceNoOutgoing}(#1)}

\newcommand{\decrSCC}{\mathit{T_d}}
\newcommand{\condense}[1]{\mathsf{CONDENSE}(#1)}
\newcommand{\set}[1]{\{#1\}}

\begin{document}

\maketitle
\begin{abstract}
	The fundamental model-checking problem, given as input a model and a specification, 
	asks for the algorithmic verification of whether the model satisfies the specification.
	Two classical models for reactive systems are graphs and Markov decision processes (MDPs).
	A basic specification formalism in the verification of reactive systems is the 
	strong fairness (aka Streett) objective, where given different types of 
	requests and corresponding grants, the requirement is that for each type, 
	if the request event happens infinitely often, then the corresponding grant 
	event must also happen infinitely often. 
	All $\omega$-regular objectives can be expressed as Streett objectives
	and hence they are canonical in verification. 
	Consider graphs/MDPs with $n$ vertices, $m$ edges, and a Streett objectives 
	with $k$ pairs, and let $b$ denote the size of the description of the Streett 
	objective for the sets of requests and grants.
	The current best-known algorithm for the problem requires time 
	$O(\min(n^2, m \sqrt{m \log n}) + b \log n)$. 
	In this work we present randomized near-linear time algorithms, with 
	expected running time $\widetilde{O}(m + b)$, where the $\widetilde{O}$ notation
	hides poly-log factors. 
	Our randomized algorithms are near-linear in the size of the input, and hence
	optimal up to poly-log factors.
\end{abstract}

\section{Introduction}\label{sec:intro}
In this work we present near-linear (hence near-optimal) randomized algorithms 
for the strong fairness verification in graphs and Markov decision processes (MDPs).
In the fundamental model-checking problem, the input is a {\em model} and a 
{\em specification}, and the algorithmic verification problem is to check whether
the model {\em satisfies} the specification. 
We first describe the models and the specifications we consider, then the notion 
of satisfaction, and then previous results followed by our contributions.

\para{Models: Graphs and MDPs.}
Graphs and Markov decision processes (MDPs) are two classical models of reactive 
systems.
The states of a reactive system are represented by the vertices of a graph, the transitions
of the system are represented by the edges and non-terminating trajectories of the system
are represented as infinite paths of the graph.
Graphs are a classical model for reactive systems with nondeterminism, and 
MDPs extend graphs with probabilistic transitions that represent reactive systems with 
both nondeterminism and uncertainty. 
Thus graphs and MDPs are the standard models of reactive systems with nondeterminism,
and nondeterminism with stochastic aspects, respectively~\cite{ClarkeBook,baierbook}.
Moreover MDPs are used as models for concurrent finite-state processes~\cite{CourcoubetisY95,Vardi85}
as well as probabilistic systems in open environments~\cite{Segala95,PRISM,STORM,baierbook}.

\para{Specification: Strong fairness (aka Streett) objectives.}
A basic and fundamental specification formalism in the analysis of reactive systems 
is the {\em strong fairness condition}.
The strong fairness conditions (aka Streett objectives) consist of $k$ types 
of requests and corresponding grants, and the requirement is that for each type if the 
request happens infinitely often, then the corresponding grant must also happen
infinitely often. 
Beyond safety, reachability, and liveness objectives, the most standard properties
that arise in the analysis of reactive systems are Streett objectives,  
and chapters of standard textbooks in verification are devoted to it (e.g., 
\cite[Chapter~3.3]{ClarkeBook},~\cite[Chapter~3]{MPProgress},~\cite[Chapters~8,~10]{AH04}).
In addition, $\omega$-regular objectives can be specified as Streett objectives, e.g., 
LTL formulas and non-deterministic $\omega$-automata can be translated to
deterministic Streett automata~\cite{Safra88} and efficient translations
have been an active research area~\cite{ChatterjeeGK13,EsparzaK14,KomarkovaK14}. 
Consequently, Streett objectives are a canonical class of objectives that arise in verification.

\para{Satisfaction.} 
The notions of satisfaction for graphs and MDPs are as follows:
For graphs, the notion of satisfaction requires that there is a trajectory (infinite path) 
that belongs to the set of paths specified by the Streett objective.
For MDPs the satisfaction requires that there is a strategy to resolve the nondeterminism 
such that the Streett objective is ensured almost-surely (with probability~1).
Thus the algorithmic model-checking problem of graphs and MDPs with Streett objectives 
is a central problem in verification, and is 
at the heart of many state-of-the-art tools such as SPIN, NuSMV for graphs~\cite{SPIN,NUSMV}, 
PRISM, LiQuor, Storm for MDPs~\cite{PRISM,LIQUOR,STORM}.

Our contributions are related to the algorithmic complexity of graphs and MDPs with 
Streett objectives. 
We first present previous results and then our contributions.

\subparagraph*{Previous results.}
The most basic algorithm for the problem for graphs is based on repeated SCC (strongly
connected component) computation, and informally can be described as follows:
for a given SCC, (a)~if for every request type that is present in the SCC 
the corresponding grant type is also present in the SCC, then the SCC is identified 
as ``good'', (b)~else vertices of each request type that have no corresponding
grant type in the SCC are removed, and the algorithm recursively proceeds
on the remaining graph.
Finally, reachability to good SCCs is computed. 
The algorithm for MDPs is similar where the SCC computation is replaced with 
maximal end-component (MEC) computation, and reachability to good SCCs is replaced
with probability~1 reachability to good MECs. 
The basic algorithms for graphs and MDPs with Streett objective have been improved in several
works, such as for graphs in~\cite{HT96,ChatterjeeHL15}, for MEC computation 
in~\cite{CH11,ChatterjeeH12,CH2014}, and MDPs with Streett objectives in~\cite{CDHL16}.
For graphs/MDPs with $n$ vertices, $m$ edges, and $k$ request-grant pairs with $b$ denoting 
the size to describe the request grant pairs, the current best-known bound is $O(\min(n^2, m \sqrt{m
\log n}) + b \log n)$. 

\subparagraph*{Our contributions.} 
In this work, our main contributions are randomized near-linear time (i.e.\ linear times a polylogarithmic factor) algorithms 
for graphs and MDPs with Streett objectives. In detail, our contributions are as follows:

\begin{itemize}

\item First, we present a near-linear time randomized algorithm for graphs with Streett objectives
	where the expected running time is $\widetilde{O}(m + b)$, where the $\widetilde{O}$ notation
	hides poly-log factors. 
	Our algorithm is based on a recent randomized algorithm for maintaining the SCC decomposition of graphs 
	under edge deletions, where the expected total running time is near linear~\cite{BPW19}.

\item Second, by exploiting the results of~\cite{BPW19} we present a randomized near-linear time
	algorithm for computing the MEC decomposition of an MDP where the expected running time is $\widetilde{O}(m)$.
	We extend the results of~\cite{BPW19} from graphs to MDPs and present a randomized algorithm to maintain 
	the MEC decomposition of an MDP under edge deletions, where the expected total running time is near
	linear~\cite{BPW19}.

\item Finally, we use the result of the above item to present a near-linear time randomized algorithm 
	for MDPs with Streett objectives where the expected running time is $\widetilde{O}(m + b)$.

\end{itemize}
All our algorithms are randomized and since they are near-linear in the size of the input, they are optimal 
up to poly-log factors. 
An important open question is whether there are deterministic algorithms that can improve the 
existing running time bound for graphs and MDPs with Streett
objectives. Our algorithms are deterministic except for the invocation of the decremental SCC
algorithm presented in~\cite{BPW19}.

\begin{table}[h!]
		\caption{Summary of Results.}\label{tab:results}
	\centering
		\renewcommand{\arraystretch}{1.10}
	\begin{tabular}{|l|l|l|}
		\hline
		
		Problem & New Running Time & Old Running Time \\ \hline
		Streett Objectives on Graphs & $\O(m + b)$ & $\O(\min(n^2, m \sqrt{m}) +
		b)~\cite{CHL17,HT96}$ \\
		Almost-Sure Reachability & $\O(m)$ & $O(m \cdot n^{2/3})$~\cite{CDHL16,CH2014}  \\
		MEC Decomposition & $\O(m)$ & $O(m\cdot n^{2/3})$~\cite{CH2014}  \\
		Decremental MEC Decomposition & $\O(m)$ & $O(nm)$~\cite{CH2014}\\
		Streett Objectives on MDPs   & $\O(m+b)$ & $\O(\min(n^2, m \sqrt{m}) +
		b)~\cite{CDHL16}$\\
	
		\hline
	\end{tabular}
\end{table}
\section{Preliminaries}

A {\em Markov decision process (MDP)} $P\! =\! ((V,E), \ls V_1, V_R \rs,
\delta)$ 
consists of a finite set of vertices $V$ partitioned into 
the player-1 vertices $V_1$ and the random vertices $V_R$, 
a finite set of edges $E \subseteq (V \times V)$, 
and a probabilistic transition function $\delta$.  
The probabilistic transition function maps every random
vertex in $V_R$ to an element of $\D(V)$, where $\D(V)$ is the set of
probability distributions over the set of vertices $V$. A random vertex $v$ has
an edge to a vertex $w \in V$, i.e.\ $(v,w) \in E$ iff
$\delta(v)[w] > 0$. An edge $e = (u,v)$ is a \emph{random edge} if $u \in V_R$ otherwise it is a \emph{player-1	edge}.
W.l.o.g. we assume $\delta(v)$ to be the uniform distribution over vertices $u$ with $(v,u) \in E$.

\emph{Graphs} are a special case of MDPs with $V_R = \emptyset$.
The set $\In{v}$ ($\Out{v}$) describes the set
of predecessors (successors) of a vertex $v$.
More formally, $\In{v}$ is defined as the set $\{w \in V \mid (w,v) \in E \}$
and $\Out{v} = \{ w \in V \mid (v,w) \in E \}$.
When $U$ is a set of vertices, we define $\edgeset{U}$
to be the set of all edges incident to the vertices in $U$. More
formally, $\edgeset{U} = \{(u,v) \in E \mid u \in U	\lor v \in U\}$. 
With $G[S]$ we denote the subgraph of a graph $G = (V,E)$ induced by the set of vertices $S
\subseteq V$. Let $\GraphReach{S}$ be the set of vertices in $G$ that can reach
a vertex of $S \subseteq V$. The set $\GraphReach{S}$ can be found in linear
time using depth-first search~\cite{tarjan1972depth}. When a vertex $u$ can
reach another vertex $v$ and vice versa, we say that $u$ and $v$ are \emph{strongly
connected}.

A {\em play} is an infinite sequence $\omega = \ls v_0, v_1, v_2, \dots \rs$ of
vertices such that each $(v_{i-1},v_i) \in E$ for all $i \geq 1$. 
The set of all plays is denoted with $\Omega$.
A play is initialized by placing a token on an initial vertex. 
If the token is on a vertex owned by player-1, he moves the token along one of the
outgoing edges, whereas if the token is at a random vertex 
$v \in V_R$, the next vertex is chosen according to the probability 
distribution $\delta(v)$.
The infinite sequence of vertices (infinite walk) formed in this way is a play.

{\em Strategies} are recipes for player~1 to extend finite prefixes of plays.
Formally, a player-1 \emph{strategy} is a function $\sigma: V^* \cdot V_1 \mapsto V$ 
which maps every finite prefix $\omega \in V^* \cdot V_1$ of a play that ends in a 
player-1 vertex $v$ to a successor vertex $\sigma(\omega) \in V$, i.e., 
$(v, \sigma(\omega)) \in E$.
A player-1 strategy is \emph{memoryless} if $\sigma_1(\omega) = \sigma_1(\omega')$ 
for all $\omega, \omega' \in V^* \cdot V_1$ that end in the same vertex $v \in V_i$, i.e., 
the strategy does not depend on the entire prefix, but only on the last vertex. 
We write $\Sigma$ for the set of all strategies for player~1.

The \emph{outcome of strategies} is defined as follows:
In graphs, given a starting vertex, a strategy for player~1 induces a unique
play in the graph. 
In MDPs, given a starting vertex $v$ and a strategy $\sigma \in \Sigma$, 
the outcome of the game is a random walk $w^\sigma_v$ for which the
probability of every event is uniquely defined, where an
\emph{event} $\A \subseteq \Omega$ is a  measurable set of plays~\cite{Vardi85}. For a vertex
$v$, strategy $\sigma \in \Sigma$ and an event $\A \subseteq \Omega$, we
denote by $\Pr^\sigma_v(\A)$ the probability that a play belongs to
$\A$ if the game starts at $v$ and player~1 follows $\sigma$.

An \emph{objective} $\phi \subseteq \Omega$ for player~1 is an event, i.e.,
objectives describe the set of winning plays.
A play $\omega \in \Omega$ \emph{satisfies} the objective if $\omega \in
\phi$. In MDPs, a player-1 strategy $\sigma \in
\Sigma$ is almost-sure (a.s.) winning from a starting vertex $v \in V$ for an objective $\phi$ iff
$\Pr_v^\sigma(\phi) = 1$. The \emph{winning set} $\ASW{\phi}$
for player~1 is the set of vertices from which player~1 has an almost-sure
winning strategy. We consider Reachability objectives and $k$-pair Streett objectives.

Given a set $T \subseteq V$ of vertices, the
\emph{reachability objective} $\Reach{T}$ requires that
some vertex in $T$ be visited. Formally, the sets of winning plays are
$\Reach{T} = \{ \ls v_0, v_1, v_2, \dots \rs \in \Omega \mid \exists k \geq 0 \text{ s.t. } v_k	\in T \}$. 
We say $v$ can reach $u$ almost-surely (a.s.) if $v \in \ASW{\Reach{\{u\}}}$.

The \emph{$k$-pair Streett objective} consists of $k$-Streett
pairs\\ $(L_1, U_1), (L_2 U_2), \dots, (L_k,U_k)$ where all $L_i, U_i \subseteq V$
for $1 \leq i \leq k$. An infinite path satisfies the objective iff for all $1
\leq i \leq k$ some vertex
of $L_i$ is visited infinitely often, then some vertex of $U_i$ is visited
infinitely often.

Given an MDP $P = (V,E, \langle V_1, V_R \rangle,\delta)$, an {\em end-component} is a set of vertices $X \subseteq V$
s.t. (1)~the subgraph induced by $X$ is strongly connected (i.e., $(X,E \cap X \times
X)$ is strongly connected) and (2)~all random vertices have their outgoing edges in $X$. More
formally, for all $v \in X \cap V_R$ and all $(v,u) \in E$ we have $u \in X$.
In a graph, if (1) holds for a set of vertices $X \subseteq V$
we call the set $X$ strongly connected subgraph (SCS). 
An end-component, SCS respectively, is \emph{trivial} if it only
contains a single vertex with no edges. All other end-components, SCSs respectively, are \emph{non-trivial}.
A \emph{maximal end-component} (MEC) is an end-component
which is maximal under set inclusion. The importance of MECs is as follows:
(i)~it generalizes \emph{strongly connected components} (SCCs) in graphs (with
$V_R = \emptyset$) and closed recurrent sets of Markov chains (with
$V_1=\emptyset$); and (ii)~in a MEC $X$, player-1 can almost-surely reach all vertices $u
\in X$ from every vertex $v \in X$. 
The MEC-decomposition of an MDP is the partition of the vertex set into MECs and the set of vertices which do
not belong to any MEC\@. 
The condensation of a graph $G$ denoted by $\condense{G}$ is the graph where all
vertices in the same SCC in $G$ are contracted. The vertices of
$\condense{G}$ are called \emph{nodes} to distinguish them from the vertices in $G$.

Let $C$ be a strongly connected component (SCC) of $G = (V,E)$.
The SCC $C$ is a \emph{bottom SCC} if no vertex $v \in C$ has an edge to a
vertex in $V\setminus C$, i.e., no outgoing edges.
Consider an MDP $P = (V,E,\langle V_1, V_R \rangle, \delta)$ and notice 
that every bottom SCC $C$ in the graph $G = (V,E)$ is a MEC because no vertex (and thus no random
vertex) has an
outgoing edge. 

A \emph{decremental graph algorithm} allows the deletion of player-1 edges while maintaining the
solution to a graph problem. It usually allows three kinds of
operations: (1) \emph{preprocessing}, which is computed when the initial
input is first received,
(2)~\emph{delete}, which deletes a player-1 edge and updates the data structure, and
(3)~\emph{query}, which computes the answer to the problem. The \emph{query time} is the time
needed to compute the answer to the query.
The \emph{update time} of a decremental algorithm 
is the cost for a \emph{delete} operation.
We sometimes refer to the delete operations as \emph{update operation}.
The running time of a decremental algorithm is characterized by the \emph{total update time}, i.e., 
the sum of the update times over the entire sequence of	deletions. 
Sometimes a decremental algorithm is randomized and assumes an \emph{oblivious adversary} who
fixes the sequence of updates in advance. When we use a decremental algorithm which assumes
such an oblivious adversary as a subprocedure the sequence of deleted edges must not depend on the 
random choices of the decremental algorithm.

\section{Decremental SCCs}\label{subsec:decsccs}

We first recall the result about decremental strongly
connected components maintenance in~\cite{BPW19} (cf.\ Theorem~\ref{thm:decscc} below) and then augment the result for our
purposes. 

\begin{theorem}[Theorem 1.1 in \cite{BPW19}]\label{thm:decscc}
	Given a graph $G=(V,E)$ with $m$ edges and $n$ vertices, we can maintain a
	data structure $\A$ that supports the operations
	\begin{itemize}
	\item $\delete{u,v}$: Deletes the edge $(u,v)$ from the graph
		$G$.
	\item $\query{u,v}$: Returns whether $u$ and $v$ are in the same
		SCC in $G$,
	\end{itemize}
	in total expected update time $O(m \log^4 n)$ and with worst-case
	constant query time. The bound holds against an oblivious 
	adversary.
\end{theorem}

The preprocessing time of the algorithm is $O(m + n)$ using~\cite{tarjan1972depth}.
To use this algorithm we extend the query and update operations 
with three new operations described in Corollary~\ref{cor:returnnewsccs}.

\begin{corollary}\label{cor:returnnewsccs}
	Given a graph $G = (V,E)$ with $m$ edges and $n$ vertices, we can maintain a
	data structure $\A$ that supports the operations
	\begin{itemize}
	\item $\rep{u}$ (query-operation): 
		Returns a reference to the SCC containing the vertex $u$.
	\item $\deleteannounce{E}$ (update-operation): 
		Deletes the set $E$ of edges from the graph $G$. 
		If the edge deletion creates new SCCs $C_1, \dots, C_k$ the
		operation returns a list $Q = \{ C_1, \dots, C_k \}$ 
		of references to the new SCCs.
	\item $\deleteannouncenooutgoing{E}$ (update-operation): 
		Deletes the set $E$ of edges from the graph $G$. 
		The operation returns a list $Q = \{C_1, \dots, C_k\}$ of
		references to all new SCCs with no outgoing edges.
	\end{itemize}
	in total expected update time $O(m \log^4 n)$ and worst-case
	constant query time for the first operation. 
	The bound holds against an oblivious adaptive adversary.
\end{corollary}
The first function is available in the algorithm described in~\cite{BPW19}.
The second function can be implemented directly from the construction of the data structure
maintained in~\cite{BPW19}.
The key idea for the third function is that when an SCC splits, we
consider the new SCCs. We distinguish between the largest of them and the others which we call small
SCCs. We then consider all edges incident to the small SCCs:
Note that as the new outgoing edges in the large SCC are also incident to a small 
SCC we can also determine the outgoing edges of the large SCC.
Observe that whenever an SCC splits all the small SCCs are at most half the size of the original SCC. 
That is, each vertex can appear only $O(\log n)$ times in small SCCs during the whole algorithm.
As an edge is only considered if one of the incident vertices is in a small SCC each 
edge is considered $O(\log n)$ times and the additional running time is bounded by $O(m \log n)$.
Furthermore, we define $\decrSCC$ as the running time of
the best decremental SCC algorithm which supports the operations in 
Corollary~\ref{cor:returnnewsccs}. Currently, $\decrSCC = O(m \log^4 n)$.

\section{Graphs with Streett Objectives}\label{sec:streettgraph}
In this section, we present an algorithm which computes the winning regions for graphs with Streett objectives.
The input is a directed graph $G=(V,E)$ and $k$ Streett pairs $(L_j,U_j)$ for $j
= 1,\dots, k$. The size of the input is measured in terms of $m = |E|$, $n =
|V|$, $k$ and $b = \sum_{j=1}^k(|L_j| + |U_j|) \leq 2nk$. 

\subparagraph*{Algorithm $\Streett$ and good component detection.}
Let $C$ be an SCC of $G$. In the good component detection problem, we compute (a) a
non-trivial SCS $G[X] \subseteq C$ induced by the set of vertices $X$, such that
for all $1 \leq j \leq k$ either $L_j \cap X = \emptyset$ or $U_j \cap X
\neq \emptyset$ or (b) that no such SCS exists. In the first case, there exists
an infinite path that eventually stays in $X$ and satisfies the Streett
objective, while in the latter case, there exists no path which satisfies the
Streett objective in $C$. From the results of~\cite[Chapter 9, Proposition
9.4]{AH04} the following algorithm,
called Algorithm $\Streett$, suffices for the winning set computation:
\begin{enumerate}
	\item Compute the SCC decomposition of the graph;
	\item For each SCC $C$ for which the good component detection returns an
		SCS, label the SCC $C$ as satisfying.
	\item Output the set of vertices that can reach a satisfying SCC as the
		winning set.
\end{enumerate}
Since the first and last step are computable in linear time, the running time of
Algorithm $\Streett$ is dominated by the detection of good components in SCCs.
In the following, we assume that the input graph is strongly connected and focus
on the good component detection.

\para{Bad vertices.}
A vertex is $\emph{bad}$ if there is some $1\leq j \leq k$ such that the vertex is in
$L_j$ but it is not strongly connected to any vertex of $U_j$. All other
vertices are good. Note that a good vertex might become bad if a vertex
deletion disconnects an SCS or a vertex of a set $U_j$. A good component is a
non-trivial SCS that contains only good vertices.

\para{Decremental strongly connected components.} 
Throughout the algorithm, we use the algorithm described in
Section~\ref{subsec:decsccs} to maintain the SCCs of a graph when deleting edges. In particular, we use
Corollary~\ref{cor:returnnewsccs} to obtain a list of the new SCCs which are
created by removing bad vertices. Note that we can `remove' a vertex
by deleting all its incident edges. Because the decremental SCC algorithm assumes an oblivious
adversary we sort the list of the new SCCs as otherwise the edge deletions performed by our algorithm would depend on the random choices of the decremental SCC algorithm.

\para{Data structure.}
During the course of the algorithm, we maintain a decomposition of the vertices
in $G = (V,E)$: We maintain a list $Q$ of certain sets $S \subseteq V$ such that every SCC of $G$ is contained in some $S$
stored in $Q$.
The list $Q$ provides two operations: $Q.\add{X}$ enqueues $X$ to $Q$; and $Q.\pull$ dequeues an arbitrary element $X$ from $Q$.
For each set $S$ in the decomposition, we store a data structure $D(S)$ in the list $Q$.
This data structure $D(S)$ supports the following operations 
\begin{enumerate}
	\item $\Construct{S}$: initializes the data structure for the set $S$
	\item $\Remove{S,B}$ updates $S$ to $S \setminus B$ for a set $B \subseteq V$ and returns
		$D(S)$ for the new set $S$.
	\item $\Bad{S}$ returns a reference to the set $\{ v \in S \mid \exists j \text{ with } v
		\in L_j \text{ and } U_j \cap S = \emptyset \}$
\item $\SCCs{S}$ returns the set of SCCs currently in $G[S]$. We implement $\SCCs{S}$ as a balanced binary search tree which allows logarithmic and updates and deletions.
\end{enumerate}
In~\cite{HT96} an implementation of this data structure with functions (1)-(3) is described that achieves
the following running times. For a set of vertices $S \subseteq V$, let
$\bits(S)$ be defined as $\sum_{j=1}^k (|S \cap L_j| + |S \cap U_j|)$.

\begin{lemma}[Lemma~2.1 in~\cite{HT96}]\label{lem:Streett_ds}
	After a one-time preprocessing of time $O(k)$, the data structure $D(S)$ can
	be implemented in time $O(bits(S) + |S|)$ for $\Construct{S}$, time
	$O(\bits(B) + |B|)$ for $\Remove{S,B}$ and constant running time
	for $\Bad{S}$.
\end{lemma}

We augment the data structure with the function $\SCCs{S}$ which runs in total
time of a decremental SCC algorithm supporting the first function in
Corollary~\ref{cor:returnnewsccs}.

\para{Algorithm Description.}
The key idea is that the algorithm maintains the list $Q$ of data structures $D(S)$ as described
above when deleting bad vertices.
Initially, we enqueue the data structure
returned by $\Construct{V}$ to $Q$. As long as $Q$ is non-empty, the algorithm repeatedly pulls a set $S$ from $Q$ and
identifies and deletes bad vertices from $G[S]$. If no edge is contained in
$G[S]$, the set $S$ is removed as it can only induce trivial SCCs. Otherwise, the
subgraph $G[S]$ is either determined to be strongly connected and output as a
good component or we identify and remove an SCC with at most
half of the vertices in $G[S]$. Consider Figure~\ref{fig:goodcomp} for an illustration of an example
run of Algorithm~\ref{alg:goodcomp}.

\begin{figure}[ht]
	\centering
	\includegraphics[width=\textwidth]{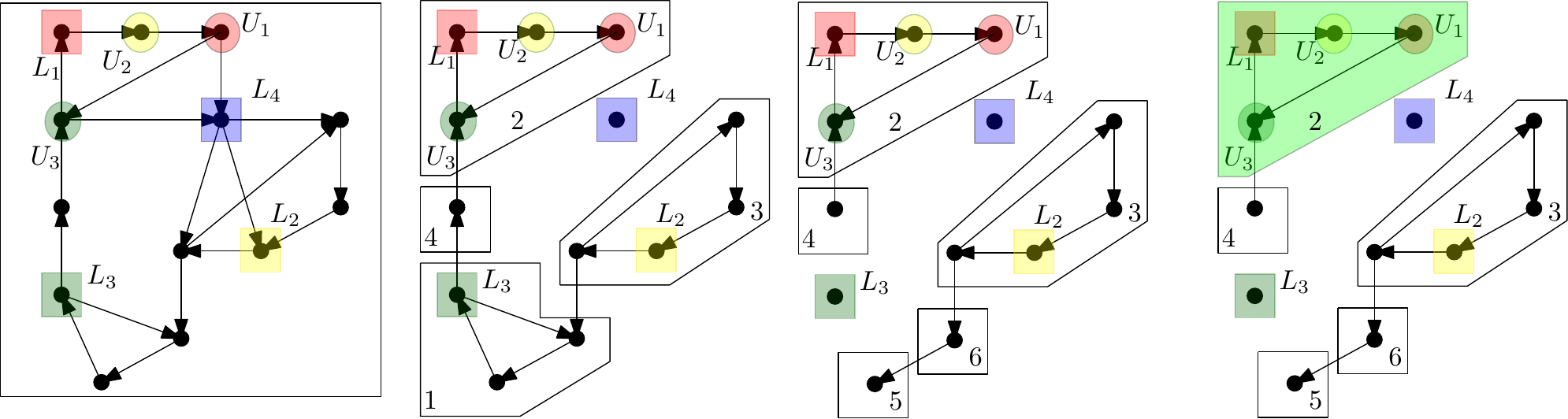}
	\caption{Illustration of one run of Algorithm~\ref{alg:goodcomp}: The vertex in the set $L_4$
		is a bad vertex and we remove it from the SCC yielding four
		new SCCs. First, we look in the SCC containing
		$L_3$. The vertex in $L_3$ is a bad vertex because there is no vertex in
		$U_3$ in this SCC\@.
		Again two SCCs are created after its removal. The next SCC we process is the SCC containing
		$L_1$. It is a good component
		because the vertex in $L_1$ has a vertex in $U_1$ in the same SCC\@. No bad vertices are
		removed and the whole SCC is identified as a good component.
}\label{fig:goodcomp}
\end{figure}

\para{Outline correctness and running time.}
In the following, 
when we talk about the \emph{input graph} $\hat{G}$ we mean the unmodified, strongly connected graph which we
use to initialize Algorithm~\ref{alg:goodcomp}. 
In contrast, with the \emph{current} graph $G$ we refer to the graph where we already deleted vertices and their incident edges in the course of finding a good component.
For the correctness of Algorithm~\ref{alg:goodcomp}, we show that if a good component exists, then
there is a set $S$ stored in list $Q$ which contains all vertices of this good component.

To obtain the running time bound of Algorithm~\ref{alg:goodcomp}, we use the fact that we can maintain the SCC decomposition
under deletions in $O(\decrSCC)$ total time. With the properties of the data
structure described in Lemma~\ref{lem:Streett_ds} we get a
running time of $\O(n+b)$ for the maintenance of the data structure and
identification of bad vertices over the whole algorithm. Combined, these ideas
lead to a total running time of $\O(\decrSCC + n + b)$ which is $\O(m + b)$ using
Corollary~\ref{cor:returnnewsccs}.

\begin{algorithm}[t]
	\small
	\caption{Algorithm~\textsc{GoodComp}}\label{alg:goodcomp}
	\KwIn{Strongly connected graph $G = (V,E)$ and Streett pairs $(L_j,U_j)$ for $j = 1, \dots, k$}
	\KwOut{a good component in $G$ if one exists}
	Invoke an instance $\A$ of the decremental SCC algorithm; Initialize $Q$ as a new list.\;
	$D(V) = \Construct{V}$; $D(V).\SCCs{V} \gets \{\A.\rep{x}\}$ for some $x \in V$\;
	$Q.\add{D(V)}$\; 
	\While{$Q$ is not empty}{\label{alg:goodcomp:outerwhile}
		$D(S) \gets Q.\pull$\;\label{alg:goodcomp:pull}
		\While{$D(S).\Bad{S}$ is not empty}{\label{alg:goodcomp:removebad1}
			$B \gets D(S).\Bad{S}$; $D(S) \gets D(S).\Remove{S,B}$\;
			\tcp{obtain SCCs after deleting bad vertices from $S$}
			$D(S).\SCCs{S} \gets D(S).\SCCs{S} \setminus \left(\bigcup_{b \in
					B} \set{\A.\rep{b}}\right)$\;  \label{alg:goodcomp:removebfromscc}
			$D(S).\SCCs{S} \gets D(S).\SCCs{S} \cup \A.\deleteannounce{\edgeset{B}}$\;\label{alg:goodcomp:dannounce}
		}\label{alg:goodcomp:removebad2}

		\If{$G[S]$ contains at least one edge}{\label{alg:goodcomp:removetrivial}
			Initialize $K$ as a new list\;
			\For{ $X \gets D(S).\SCCs{S}$ }{\label{alg:goodcomp:for1}
				\lIf{$X = S$} {
					\textbf{output} $G[S]$; \tcp*[f]{good component found}\DontPrintSemicolon\label{alg:goodcomp:outputgoodcomp} 
				}
				\lIf{$|X| \leq {|S|\over 2}$}{
					$K.\add{X}$;
				}
			}
			Sort the SCCs in $K$ by vertex id (look at all the vertices in each SCC of
			$K$)\;\label{alg:goodcomp:sort}
			$R \gets \emptyset$\tcp*{{\footnotesize Build $D(X)$ for SCCs $X$ in $K$
			and remove $X$ from $S$,$D(S)$ and $\SCCs{S}$}}
			
			\For{$X \gets K.\pull$}{\label{alg:goodcomp:for2}
				$R \gets R \cup X$; $D(X) = \Construct{X}$\;
				$D(X).\SCCs{X} \gets \set{\A.\rep{x}}$ for some $x \in
				X$\;\label{alg:goodcomp:newSmallSCC}
				$D(S).\SCCs{S} \gets D(S).\SCCs{S} \setminus \{\A.\rep{x}\}$ for some $x \in
				X$\;\label{alg:goodcomp:remSmallSCC}
				$Q.\add{D(X)}$\;\label{alg:goodcomp:ds1}
			}
			\lIf{$D(S).\SCCs{S} \not= \emptyset$}{
				$Q.\add{D(S).\Remove{S,R}}$\label{alg:goodcomp:ds2} 
			}
		}
	}

	\Return No good component exists. 
\end{algorithm}

\begin{lemma}\label{lem:goodcomp:runningtime}
	Algorithm~\ref{alg:goodcomp} runs in expected time  $\O(m + b)$.
\end{lemma}

\begin{proof}
	The preprocessing and initialization of the data structure $D$ and the removal
	of bad vertices in the whole algorithm takes time $O(m + k + b)$ using
	Lemma~\ref{lem:Streett_ds}. Since each vertex is deleted at most once, the data structure 
	can be constructed and maintained in total time $O(m)$. 
	Announcing the new SCCs after deleting the bad vertices at
	Line~\ref{alg:goodcomp:dannounce} is in $O(\decrSCC) = \O(m)$ total time by 
	Corollary~\ref{cor:returnnewsccs}.
	Consider an iteration of the while loop at Line~\ref{alg:goodcomp:outerwhile}:
	A set $S$ is removed from $Q$. Let us denote by $n'$ the number of vertices of $S$. 
	If $G[S]$ does not contain any edge	after the removal of bad vertices, 
	then $S$ is not considered further by the algorithm. 
	Otherwise, the for-loop at Line~\ref{alg:goodcomp:for1} considers all new SCCs.
	Note the we can implement the for-loop in a lockstep fashion:
	In each step for each SCC we access the $i$-th vertex and as soon as all of the vertices 
	of an SCC are accessed we add it to the list $K$.
	When only one SCC is left we compute its size using the original set $S$ and 
	the sizes of the other SCCs. If its size is at most $|S|/2$ we add it to $K$.
	Note that this can be done in time proportional to the number of vertices in the SCCs in $S$
	of size at most $|S|/2$.
	The sorting operation at Line~\ref{alg:goodcomp:sort} takes time $O(|K| \log |K|)$ plus the size of all
	the SCCs in $K$, that is $\sum_{K_i \in K} |K_i|$. Note that $O(|K| \log |K|) = O((\sum_{K_i
		\in K} |K_i|) \log (\sum_{K_i \in K} |K_i|))$.  Let $K_i \in K$ be an SCC stored in $K$. 
	Note that during the algorithm each vertex can appear at most $\log(n)$ times in the list $K$.
	This is by the fact that $K$ only contains SCCs that are at most half the size of the original set $S$.
	We obtain a running time bound
	of $O(n (\log n)^2)$ for Lines~\ref{alg:goodcomp:for1}-\ref{alg:goodcomp:sort}.\\
	Consider the second for-loop at Line~\ref{alg:goodcomp:for2}:
	Let $|X| = n_1$. The operations $\Remove{\cdot}$ and $\Construct{\cdot}$ are called once per
	found SCC $G[X]$ with $X \neq S$ and take by
	Lemma~\ref{lem:Streett_ds} $O(|X| + \bits(X))$ time. Whenever a vertex is in
	$X$, the size of the set in $Q$ containing $v$ originally is reduced by at least a factor of two due to the fact
	that $|X| = n_1 \leq n'/2$. This happens at most $\lceil \log n \rceil$
	times. By charging $O(1)$ to the vertices in $X$ and, respectively,
	to $\bits(X)$, the total running time for
	Lines~\ref{alg:goodcomp:ds1} \& \ref{alg:goodcomp:ds2} can be bounded by $O((n +
	b) \log n)$ as each vertex and bit is only charged $O(\log n)$ times.
	Combining all parts yields the claimed running time bound of $O(\decrSCC + b\log n + n \log^2n) =
	\O(m +b)$. 
\end{proof}

The correctness of the algorithm is similar to the analysis given
in~\cite[Lemmas~3.6 \& 3.7]{CHL17} except that we additionally have to prove that
$\SCCs{S}$ holds the SCCs of $G[S]$. Lemma~\ref{lem:goodcomp:ds} shows that we
maintain $\SCCs{S}$ properly for all the data structures in $Q$.

\begin{lemma}\label{lem:goodcomp:ds}
	After each iteration of the outer while-loop every non-trivial SCC of the current
	graph is contained in one of the subgraphs $G[S]$ for which the data structure $D(S)$ is
	maintained in $Q$ and $\SCCs{S}$ stores a list of all SCCs contained in $S$.
\end{lemma}

We prove the next Lemma by showing that we never remove edges of vertices of
good components.
\begin{lemma}\label{lem:goodcomp:maintaingood}
	After each iteration of the outer while-loop every good component
	of the input graph is contained in one of the subgraphs $G[S]$ for which the
	data structure $D(S)$ is maintained in the list $Q$.
\end{lemma}

\begin{proposition}\label{lem:goodcomp:correctness}
	Algorithm~\ref{alg:goodcomp} outputs a good component if one exists,
	otherwise the algorithm reports that no such component exists.
\end{proposition}
\begin{proof}
	First consider the case where Algorithm~\ref{alg:goodcomp} outputs a subgraph $G[S]$.
	We show that $G[S]$ is a good component: 
	Line~\ref{alg:goodcomp:removetrivial} ensures only non-trivial SCSs are considered. 
	After the removal of bad vertices from $S$ in Lines~\ref{alg:goodcomp:removebad1}-\ref{alg:goodcomp:removebad2}, we
	know that for all $1 \leq j \leq k$ that $U_j \cap S \neq \emptyset$ if
	$S \cap L_j \neq \emptyset$. Due to Line~\ref{alg:goodcomp:outputgoodcomp} there
	is only one SCC in $G[S]$ and thus $G[S]$ is a good component.
	Second, if Algorithm~\ref{alg:goodcomp} terminates without a good component, by
	Lemma~\ref{lem:goodcomp:maintaingood}, we have that the initial graph has 
	no good component and thus the result is correct as well.
\end{proof}

The running time bounds for the decremental SCC algorithm of~\cite{BPW19} (cf. Corollary~\ref{cor:returnnewsccs}) only
hold against an oblivious adversary. Thus we have to show that in our algorithm the sequence of edge deletions does not depend on the random choices of the decremental SCC algorithm. 
The key observation is that only the order of the computed SCCs depends on the random choices
of the decremental SCC and we eliminate this effect by sorting the SCCs.
\begin{proposition}\label{prop:goodcomp:obadv}
	The sequence of deleted edges does not depend on the random choices of the decremental SCC
	Algorithm but only on the given instance.
\end{proposition}

Due to Lemma~\ref{lem:goodcomp:runningtime}, Lemma~\ref{lem:goodcomp:correctness} and
Proposition~\ref{prop:goodcomp:obadv} we obtain the following result.
\begin{theorem}
	In a graph, the winning set for a $k$-pair Streett objective can be computed in $\O(m + b)$ expected time.  
\end{theorem}

\section{Algorithms for MDPs}\label{sec:algformdps}

In this section, we present expected near-linear time algorithms for
computing a MEC decomposition, 
deciding almost-sure reachability and
maintaining a MEC decomposition in a decremental setting. 
In the last section, we present an algorithm for \emph{MDPs} with Streett objectives by using the new algorithm
for the decremental MEC decomposition.

\para{Random attractor.}
First, we introduce the notion of a \emph{random attractor} $\attr{R}{T}$ for a set
$T \subseteq V$. The random attractor $A = \attr{R}{T}$ is defined inductively as
follows: $A_0 = T$ and $A_{i+1} = A_i \cup \{v \in V_R \mid \Out{v} \cap
	A_i \neq \emptyset \} \cup \{v \in V_{1} \mid \Out{v} \subseteq A_i 
\}$ for all $i>0$.
Given a set $T$, the random attractor includes all vertices
(1) in $T$, (2) random vertices with an edge to $A_i$, (3) player-1 vertices
with all outgoing edges in $A_i$.
Due to~\cite{I81,B80} we can compute the random attractor
$A = \attr{R}{S}$ of a set $S$ in time $O(\sum_{v \in A} \In{v})$.

\subsection{Maximal End-Component Decomposition}\label{sec:mec}
In this section, we present an expected near linear time algorithm for MEC
decomposition. 
Our algorithm is an efficient implementation of the static algorithm presented in~\cite[p.
29]{CH2014}: The difference is that the bottom SCCs are computed with a dynamic SCC algorithm instead
of recomputing the static SCC algorithm.
A similar algorithm was independently proposed in an unpublished extended version
of~\cite{CHILP16}.

\para{Algorithm Description.}
The MEC algorithm described in Algorithm~\ref{alg:mec_algo} repeatedly
removes bottom SCCs and the corresponding random attractor. 
After removing bottom SCCs the new SCC decomposition with its bottom SCCs is computed using a dynamic SCC algorithm.

\begin{algorithm}[ht]
	\caption{MEC Algorithm}\label{alg:mec_algo}
	\small
	\KwIn{MDP $P = (V, E, \langle V_1, V_R \rangle, \delta)$, decremental SCC algorithm $\A$}
	Invoke an instance $\A$ of the decremental SCC algorithm\;
	Compute the SCC-decomposition of $G = (V,E)$: $C = \{C_1, \dots, C_\ell\}$ \;
	Let $M = \emptyset$; $Q \gets \{ C_i \in C \mid \text{ $C_i$ has no outgoing edges} \}$\;
	\While{$Q$ is not empty}{
		$C \gets \emptyset$\;
		\lFor{$C_k \in Q$}{
			$C \gets C \cup C_k$; $M \gets M \cup \{C_k\}$ 
		}		
		$A \gets \attr{R}{C}$\;\label{alg:mec:attrcomp}
		$Q \gets \A.\deleteannouncenooutgoing{\edgeset{A}}$\;\label{alg:mec:removeattr}
	}
	\Return $M$\;
\end{algorithm}

Correctness follows because our algorithm just removes attractors of bottom SCCs
and marks bottom SCCs as MECs\@. This is
precisely the second static algorithm presented in~\cite[p. 29]{CH2014} except that the bottom SCCs
are computed using a dynamic data structure. 
By using the decremental SCC algorithm described in
Subsection~\ref{subsec:decsccs} we obtain the following lemma.
\begin{lemma}\label{lem:mec:running}
	Algorithm~\ref{alg:mec_algo} returns the MEC-decomposition of an MDP $P$ in expected time $\O(m)$. 
\end{lemma}
\begin{proof}
	The running time of algorithm $\A$ is in total time $O(\decrSCC) = \O(m)$ by
	Theorem~\ref{thm:decscc} and Corollary~\ref{cor:returnnewsccs}.
	Initially, computing the SCC decomposition and determining the SCCs with no
	outgoing edges takes time $O(m+n)$ by using~\cite{tarjan1972depth}. 
	Each time we compute the attractor of a bottom
	SCC $C_k$ at Line~\ref{alg:mec:attrcomp} we
	remove it from the graph by deleting all its edges and never
	process these edges and vertices again. Since we can compute the attractor $A$ at Line~\ref{alg:mec:attrcomp} in 
	time $O(\sum_{v \in A} \In{A})$, we need $O(m+n)$ total time for computing the
	attractors of all bottom SCCs.
	Hence, the running time is dominated by the decremental SCC algorithm
	$\A$, which is $O(\decrSCC) = \O(m)$.
\end{proof}

The algorithm uses $O(m+n)$ space due to the fact that the decremental
SCC algorithm $\A$ uses $O(m+n)$ space and $Q$ only contains 
vertices.
\begin{theorem}\label{thm:mec}
	Given an MDP the MEC-decomposition can
	be computed in $\O(m)$ expected time. The algorithm uses $O(m + n)$ space.
\end{theorem}

Note that we can use the decremental SCC Algorithm $\A$ of~\cite{BPW19} even though this algorithm only
works against an oblivious adversary as the sequence of deleted edges does not depend on the random choices 
of the decremental SCC Algorithm.

\subsection{Almost-Sure Reachability}
In this section, we present an expected near linear-time algorithm for the
almost-sure reachability problem. 
In the almost-sure reachability problem, we are given an MDP $P$ and a target set $T$ and
we ask for which vertices player~1 has a strategy to reach $T$ almost surely, i.e., $\ASW{\Reach{T}}$.
Due to~\cite[Theorem 4.1]{CDHL16} we can determine the set $\ASW{\Reach{T}}$ in time $O(m + \MEC)$ where $\MEC$ is the running time of the fastest MEC algorithm. 
We use Theorem~\ref{thm:mec} to compute the MEC decomposition and obtain the following theorem.

\begin{theorem}\label{thm:asreach}
	Given an MDP and a set of vertices $T$ we can compute $\ASW{\Reach{T}}$ in
	$\O(m)$ expected time.
\end{theorem}

\subsection{Decremental Maximal End-Component Decomposition}\label{sec:decmec}
We present an expected near-linear time
algorithm for the MEC-decomposition which supports player-1 edge deletions and a query 
that answers if two vertices are in the same MEC\@.
We need the following lemma from~\cite{CH11} to prove the correctness of our algorithm.
Given an SCC $C$ we consider the set U of the random vertices in $C$ with edges leaving 
$C$. The lemma states that for all non-trivial MECs $X$ in $P$ the intersection with
$U$ is empty, i.e., $\attr{R}{U} \cap X = \emptyset$.
\begin{lemma}[Lemma 2.1(1),~\cite{CH11}]\label{lem:MECcorr}
	Let $C$ be an SCC in $P$. Let $U = \set{v \in C\cap V_R \mid E(v) \cap (V \setminus C) \neq
	\emptyset}$ be the random vertices in $C$ with edges leaving $C$. Let $Z = \attr{R}{U} \cap
	C$. Then for all non-trivial MECs $X$ in $P$ we have $Z \cap X = \emptyset$ and for any edge
	$(u,v)$ with $u \in X$ and $v \in Z$, $u$ must belong to $V_1$.
\end{lemma}

The \emph{pure MDP graph} $P^P$ of an MDP $P = (V, E, \langle V_1, V_R \rangle, \delta)$ is
the graph which contains only edges in non-trivial
MECs of $P$. More formally, the pure MDP graph $P^P$ is defined as follows:
Let $M_1, \dots M_k$ be the set of MECs of $P$.
$P^P = (V^P, E^P, \langle V_1^P, V_R^P \rangle, \delta^P)$ where
$V^P = V, V_1^P = V_1, V_R^P = V_R$, 
$E^P = \bigcup_{i= 1}^k \{(u,v) \in E \cap (M_i \times M_i)\}$ and
for each $v \in V_R$: $\delta^P(v)$ the uniform distribution over vertices $u$ with $(v,u) \in E^P$.

Throughout the algorithm, we maintain the pure MDP graph $P^P$ for an input MDP
$P$. Note that every non-trivial SCC in $P^P$ is also a MEC due to the fact that there are only edges inside
of MECs. Moreover, a trivial SCC $\{v\}$ is a MEC iff $v \in V_1$.
Note furthermore that when a player-1 edge of an MDP $P$ is
deleted, existing MECs might split up into several MECs 
but no new vertices are added to existing MECs.

Initially, we compute the MEC-decomposition in $\O(m)$ expected time using the algorithm
described in Section~\ref{sec:mec}. Then we remove every edge that is not in a
MEC\@. 
The resulting graph is the pure MDP graph $P^P$. 
Additionally, we invoke a decremental SCC algorithm $\A$ which is able to (1) announce new SCCs
under edge deletions and return a list of their vertices and (2) is able to
answer queries that ask whether two vertices $v,u$ belong to the same SCC\@.
When an edge $(u,v)$ is deleted, we know that (i) the MEC-decomposition stays
the same or (ii) one MEC splits up into new MECs and the rest of the
decomposition stays the same.
We first check if $u$ and $v$ are in the same
MEC, i.e., if it exists in $P^P$. If not, we are done. Otherwise, $u$ and $v$ are
in the same MEC $C$ and either (1) the MEC $C$ does not split or (2) the MEC
$C$ splits. In the case of (1) the SCCs of the pure MDP graph $P^P$ remain intact and nothing needs to be done. In the case of (2) we need to identify the new SCCs $C_1, \dots, C_k$ in $P^P$ using the decremental SCC algorithm $\A$. Let, w.l.o.g., $C_1$ be the SCC with the most vertices. We iterate through every edge
of the vertices in the SCCs $C_2, \dots, C_k$. By considering all the edges, we identify 
all SCCs (including $C_1$) which are also MECs. We remove all edges $(y,z)$ where $y$ and
$z$ are not in the same SCC to maintain the pure MDP graph $P^P$. 
For the SCCs that are not MECs let $U$ be the set of random vertices with edges
leaving its SCC\@. We compute and remove $A= \attr{R}{U}$ (these vertices belong to no MEC due to
Lemma~\ref{lem:MECcorr}) and recursively start the procedure on the new SCCs
generated by the deletion of the attractor. The algorithm is illustrated in Figure~\ref{fig:decr_mec}.
\begin{figure}
	\centering
	\includegraphics[width=0.9\textwidth]{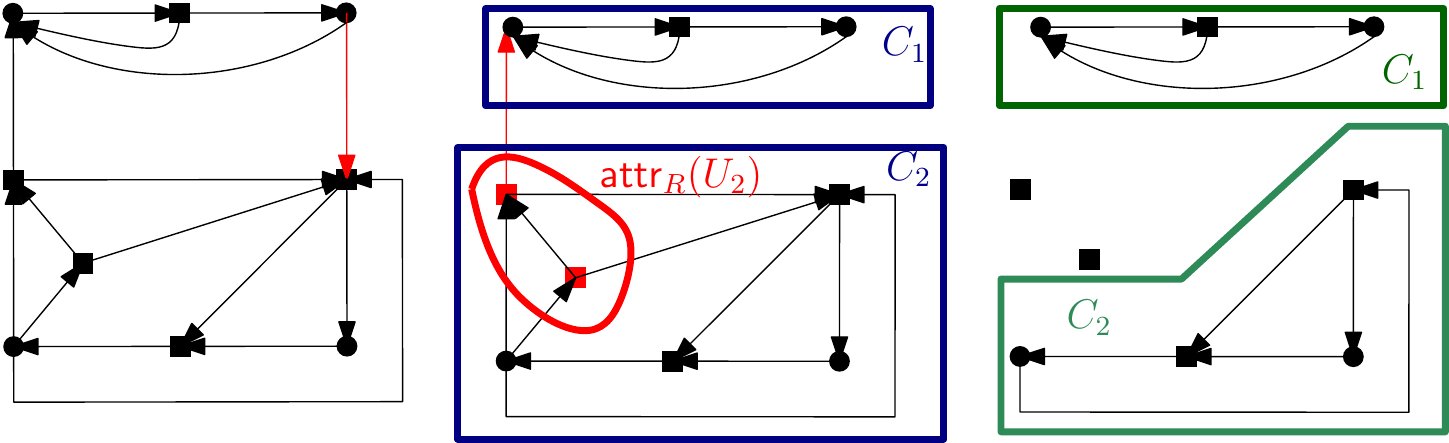}
	\caption{We delete an edge which splits the MEC into two new SCCs $C_1$ and $C_2$. The SCC $C_2$ 
	is not a MEC. We thus compute and remove the attractor of $U_2$ and the resulting SCC is a
MEC.}\label{fig:decr_mec}
\end{figure}
\begin{algorithm}[t]
	\caption{Decremental MEC-update}\label{alg:decr_mec}
	\small
	\KwIn{Player-1 Edge $e=(u,v)$}
	\If{$\A.\query{u,v}=true$}{\label{alg:decr_mec:checkid} 
		List $K \gets \{\A.\deleteannounce{(u,v)}$\}\;\label{alg:decr_mec:removeedge}
		\While{$K \neq \emptyset$}{\label{alg:decr_mec:newsccs} 
			pull a list $J$ of SCCs from $K$\ and let $C_1$ be the largest SCC\;
			$\{C_1, \dots C_k\} \gets$ Sort all SCCs in $J$ except $C_1$ by the smallest vertex id.\;\label{alg:decr_mec:sort}
			$\MEC^{C_1} = \True$, $U_1 \gets \emptyset$\;\label{alg:decr_mec:labelc1true}

			\For{$i = 2;\ i \leq k;\ i\!+\!+$}{\label{alg:decr_mec:iteratesmaller}
				$\MEC^{C_i} = \True$, $U_i \gets \emptyset$\;\label{alg:decr_mec:labelcitrue}

				\For{$e=(s,t)$ where $e \in
					\edgeset{C_i}$}{\label{alg:decr_mec:iterateedges2k}
					\lIf{$(s \notin C_i) \lor (t \notin	C_i)$}{\label{alg:decr_mec:putInD}
						$\A.\delete{e}$
					}

					\lIf{$(s \in V_R \land t \notin
						C_i)$}{\label{alg:decr_mec:labelcifalse}
						$\MEC^{C_i} = \False$; $U_i \gets U_i \cup \set{s}$
					}

					\lIf{$(s \in V_R \land s \in
						C_1)$}{\label{alg:decr_mec:labelc1false}
						$\MEC^{C_1} = \False$; $U_1 \gets U_1 \cup \set{s}$
					}
				}\label{alg:decr_mec:iteratesmaller_end}
				\uIf{$\MEC^{C_i} = \False$}{\label{alg:decr_mec:setMECi}
					$A \gets \attr{R}{U_i}\cap
					C_i$\;\label{alg:decr_mec:attrci}
					$J \gets \A.\deleteannounce{\edgeset{A}}\setminus \left(\bigcup_{a \in A}\A.\rep{a}\right)$\;\label{alg:decr_mec:removeattrci}
					\lIf{$ J \neq \emptyset$}{
						$K \gets K \cup \{J\}$\label{alg:decr_mec:addKi}
					}

				}
			}
			\If{$\MEC^{C_1} = \False$}{\label{alg:decr_mec:setMEC1}
				$A \gets \attr{R}{U_1} \cap C_1$\; \label{alg:decr_mec:attrc1}
				$J \gets \A.\deleteannounce{\edgeset{A}}\setminus \left(\bigcup_{a \in A}\A.\rep{a}\right)$\;\label{alg:decr_mec:removeattrc1}
				\lIf{$ J \neq \emptyset$}{
					$K \gets K \cup \{J\}$\label{alg:decr_mec:addK1}
				}
			}		
		} 
	}
\end{algorithm}

\noindent Lemma~\ref{lem:decr_mec:inv2} describes the key invariants of the while-loop at
Line~\ref{alg:decr_mec:newsccs}. 
We prove it with a straightforward induction on the number
of iterations of the while-loop and apply Lemma~\ref{lem:MECcorr}.
\begin{lemma}\label{lem:decr_mec:inv2}
	Assume that $\A$ maintains the pure MDP graph $P^P$
	before the deletion of $e = (u,v)$ then the while-loop at Line~\ref{alg:decr_mec:newsccs} maintains the
	following invariants: 
	\begin{enumerate}
	\item For the graph stored in $\A$  and all lists of SCCs $\{C_1, \dots, C_k\}$ in $K$ there are only edges inside the
		SCCs or between the SCCs in the list, i.e., 
		for each $(x,y) \in \bigcup_{j=0}^k E[C_j]$ we have $x,y \in \bigcup_{j=0}^k C_j$.
	\item If a non-trivial SCC of the graph in $\A$ is not a MEC of the current MDP it is in $K$.
	\item If $M$ is a MEC of the current MDP then we do not delete an edge of $M$ in the
		while-loop. 
	\end{enumerate}
\end{lemma}

\begin{proposition}\label{prop:decrMECcorrectness}
	Algorithm~\ref{alg:decr_mec} maintains the pure MDP graph $P^P$ in the data structure $\A$ under
	player-1 edge deletions.
\end{proposition}
\begin{proof}
	We show that after deleting an edge using Algorithm~\ref{alg:decr_mec}
	(i) every non-trivial SCC is a MEC and vice-versa, and (ii) there are no edges going from one MEC to another.
	Initially, we compute the pure MDP graph and both conditions are fulfilled.

	When we delete an edge and the while-loop at Line~\ref{alg:decr_mec:newsccs}
	terminates (i) is true due to Lemma~\ref{lem:decr_mec:inv2}(2,3). 
	That is, as we never delete edges within MECs they are still strongly
	connected  and
	when the while-loop terminates, $K = \emptyset$ which means that all SCCs
	are MECs. 	

	For (ii) notice that each SCC is once processed as a List $J$. 
	Consider an arbitrary SCC $C_i$ and the corresponding list of SCCs $J = \{C_1, \dots, C_k\}$ of the
	iteration in which $C_i$ was identified as a MEC\@. By
	Lemma~\ref{lem:decr_mec:inv2}(1) there are no edges to SCCs not in the list. 
	Additionally, due to Line~\ref{alg:decr_mec:putInD} we remove all edges from $C_i$ to other SCCs
	in $J$.
\end{proof}

Now that we maintain the pure MDP graph $P^P$ in $\A$,
we can answer MEC queries of the form:
$\query{u,v}$: Returns whether $u$ and $v$ are in the same MEC in $P$, by an SCC query
$\A.\query{u,v}$ on the pure MDP graph $P^P$. 

The key idea for the running time of Algorithm~\ref{alg:decr_mec} is that
we do not look at edges of the largest SCCs but the new SCC decomposition by inspecting the edges of
the smaller SCCs. Note that we identify the largest SCC by processing the SCCs in a lockstep manner. This can only happen $\lceil \log n
\rceil$ times for each edge. Additionally, when we sort the SCCs, we only
look at the vertex ids of the smaller SCCs and when we charge this cost to the
vertices we need $O(n \log^2 n)$ additional time.

\begin{proposition}\label{prop:decrMECrunning}
	Algorithm~\ref{alg:decr_mec} maintains the MEC-decomposition of	$P$ under
	player-1 edge deletions in expected total time $\O(m)$. Algorithm~\ref{alg:decr_mec} answers queries
	that ask whether two vertices $v,u$ belong to the same MEC in $O(1)$. The
	algorithm uses $O(m + n)$ space.
\end{proposition}

Due to the fact that the decremental SCC algorithm we use in Corollary~\ref{cor:returnnewsccs} only
works for an oblivious adversary, we prove the following proposition. The key
idea is that we sort SCCs returned by the decremental SCC Algorithm. Thus, 
the order in which new SCCs are returned does only depend on the given instance.
\begin{proposition}\label{prop:mec_decr_oblivious}
	The sequence of deleted edges does not depend on the random choices of the
	decremental SCC Algorithm but only on the given instance.
\end{proposition}

The algorithm presented in~\cite{BPW19} fulfills all the conditions of
Proposition~\ref{prop:decrMECrunning} due to Corollary~\ref{cor:returnnewsccs}. Therefore we obtain the following theorem due to
Proposition~\ref{prop:decrMECcorrectness} and
Proposition~\ref{prop:decrMECrunning}.

\begin{theorem}\label{thm:mec_decr}
	Given an MDP with $n$ vertices and $m$ edges, the MEC-decomposition can be maintained under the deletion of $O(m)$ player-1 edges
	in total expected time $\O(m)$ and we can answer queries that ask whether two vertices $v,u$ belong to the same MEC in $O(1)$ time. 
	The	algorithm uses $O(m + n)$ space. The bound holds against an oblivious
	adversary.
\end{theorem}

\subsection{MDPs with Streett Objectives}

Similar to graphs we compute the winning region of Streett objectives with $k$ pairs $(L_i,U_i)$ (for $1
\leq i \leq k$) for an MDP $P$ as follows: 
\begin{enumerate}
\item We compute the MEC-decomposition of $P$.
\item For each MEC, we find good end-components, i.e., end-components where
	$L_i \cap X = \emptyset$ or $U_i \cap X \neq \emptyset$ for all $1 \leq i \leq k$ and label the MEC as satisfying.
\item We output the set of vertices that can almost-surely reach a satisfying MECs.
\end{enumerate}
For 2., we find good \emph{end-components} similar to how we find good
components as in Section~\ref{sec:streettgraph}.
The key idea is to use the decremental MEC-Algorithm described in
Section~\ref{sec:decmec} instead of the decremental SCC Algorithm. 
We modify the Algorithm presented in Section~\ref{sec:streettgraph} as follows
to detect good end-components:
First, we use the decremental MEC-algorithm instead of the decremental SCC Algorithm.
Towards this goal, we augment the decremental MEC-algorithm with a function to return a list
of references to the new MECs when we delete a set of edges. 
Second, the decremental MEC-algorithm does not allow the deletion of arbitrary edges, but only player-1 edges. 
To overcome this obstacle, we create an equivalent instance where we remove player-1 edges when we remove `bad' vertices. 
\begin{lemma}\label{cor:returnnewmecs}
	Given an MDP $P = (V,E,\langle V_1, V_R \rangle, \delta)$ with $m$ edges and $n$ vertices, we can maintain a
	data structure that supports the operation
	\begin{itemize}
	\item $\deleteannounce{E}$: Deletes the set of $E$ of player-1 edges $(u,v)$ from
		the MDP $P$. If the edge deletion creates new MECs $C_1, \dots, C_k$ the
		operation returns a list $Q = \{ C_1, \dots, C_k \}$ of references
		to the new non-trivial MECs.
	\end{itemize}
	in total expected update time $\O(m)$. The bound holds against an oblivious adaptive adversary.
\end{lemma}
\para{Deleting bad vertices.}
As the decremental MEC-algorithm only allows deletion of player-1 edges, we first modify the original instance $P = (V,E,\langle V_1,V_R \rangle, \delta)$ to a new instance $P' =(V',E',\langle V_1',V_R' \rangle, \delta')$ such that we can remove bad vertices by
deleting player-1 edges only.
In $P'$ each vertex $v \in V_x$ for $x \in \{1, R\}$
is split into two vertices $v_{in} \in V_1'$ and $v_{out} \in V_x'$ such that 
$E'=\{ (u_{out},v_{in}) \mid (u,v) \in E\} \cup \{(v_{in}, v_{out})\mid v \in V\}$
and $L_i' = \{ v_{in}  \in V' \mid v \in L_i \}$ and $U_i' = \{ v_{out} \in V' \mid v \in U_i \}$ for all $1 \leq i \leq k$.  
The new probability distribution is $\delta'(v_{out})[w_{in}] = \delta(v)[w]$ for $v \in V_R$ and
$w \in \Out{v}$. 
Note that for each $v \in V_R$ the corresponding vertex $v_{out} \in
V_R' $ has the same probabilities to reach the representation $v_{out}$ of a vertex as $v$.
The described reduction allows us to remove bad vertices from MECs by removing the
player-1 edge $(v_{in},v_{out})$. 

The key idea for the following lemma is that for each original vertex $v \in V$
either both $v_{in}$ and $v_{out}$ are part of a good end-component or none of
them. 
Note that the only way that $v_{in}$ and $v_{out}$ are strongly connected is when the other
vertex is also in the strongly connected component because $v_{in}$ ($v_{out}$) has only one
outgoing (incoming) edges to $v_{out}$ (from $v_{in}$).

\begin{lemma}\label{lem:streetmec:modifiedinstance}
	There is a good end-component in the modified instance $P'$ iff there is a good
	component in the original instance $P$.
\end{lemma}

\noindent On the modified instance $P'$ the algorithm for MDPs is identical to Algorithm~\ref{alg:goodcomp} except that
we use a dynamic MEC algorithm instead of a dynamic SCC algorithm.

\begin{theorem}\label{thm:mdp_Streett}
	In an MDP the winning set for a $k$-pair Street objectives can be computed in $\O(m + b)$ expected time. 
\end{theorem}

\bibliography{streett}
\end{document}